\documentclass[english,11pt]{article}
\input{macros}

\begin{document}

\title{\Large Learning a Size--Weight Frontier for Synthetic-Augmented Inference}

\author{
	Chengpiao Huang\footnote{Department of IEOR, Columbia University. Email: \href{mailto:chengpiao.huang@columbia.edu}{chengpiao.huang@columbia.edu}.}
	\and Kaizheng Wang\footnote{Department of IEOR and Data Science Institute, Columbia University. Email: \href{mailto:kaizheng.wang@columbia.edu}{kaizheng.wang@columbia.edu}.}
}

\date{\today}

\maketitle

\begin{abstract}
Synthetic data can improve statistical inference when real data are scarce, but na\"{i}vely treating synthetic samples as real data can introduce bias and lead to unreliable inference. We develop a general framework for synthetic-augmented inference across a population of related tasks. It characterizes synthetic augmentation by the number of synthetic observations and their weight. Central to our framework is a size--weight frontier that specifies, for each weight, the largest synthetic sample size for which all smaller sizes attain the target task-marginal coverage. We estimate this frontier from historical tasks, and establish a finite-sample coverage guarantee simultaneously for all size--weight configurations on or below the estimated frontier. In experiments using large language model responses to augment opinion survey data, our procedure achieves target coverage and substantially narrows confidence intervals.
\end{abstract}


\section{Introduction}

Synthetic data are increasingly used to augment real data when they are scarce, costly, or slow to collect. For example, a large language model can generate many responses to a survey question at low cost. While combining synthetic data with a small human sample can reduce estimation variance, the synthetic simulator may systematically deviate from the human population. As a consequence, na\"{i}vely treating synthetic samples as real data can lead to unreliable inference: the resulting confidence intervals may be narrower but fail to maintain valid statistical coverage. 

In this paper, we study how to control the contribution of synthetic data when constructing confidence sets. We consider a population of related tasks, where a task may be a survey question, a customer segment, or a newly arrived user. For the target task, our goal is to construct a confidence set for a task-specific parameter with at least a prescribed coverage probability. We observe $n$ real responses and may generate $k$ synthetic responses. We assign weight one to each real sample, and a nonnegative weight $\weight$ to each synthetic sample. Thus, $k$ controls the number of synthetic samples, and $\weight$ controls their individual influence. Our goal is to determine:
\begin{center}
\emph{For each weight $\weight$, what sample size $k$ can we use while maintaining valid confidence set coverage?}
\end{center}

The valid configurations $(\weight,k)$ depend on the discrepancy between the real and synthetic data distributions, which is generally unknown and may vary across tasks. To learn these configurations, we assume access to historical tasks drawn from the same task population as our target task, and for which more real samples are available. On each historical task, we use the additional real data to construct a reference confidence set. The reference set is then compared with the synthetic-augmented confidence sets to evaluate and calibrate different configurations $(\weight,k)$. This setting captures cold-start problems, including inference for new users or items in recommendation systems, where established tasks with more real data guide how much synthetic information to augment for a new task with limited real data.

\paragraph{Main contributions.}
We introduce a \emph{size--weight frontier} as the central object for calibrating synthetic augmentation in statistical inference. For each weight in a candidate set, the frontier gives the largest synthetic sample size for which all smaller sizes lead to confidence sets with valid coverage. We then develop data-driven procedures that use historical tasks to learn this frontier, and establish a finite-sample coverage guarantee that simultaneously holds for every configuration on or below the learned frontier; see \Cref{fig-frontier-illustration} for an illustration. The uniform coverage allows us to select a configuration after calibration, e.g., by minimizing average confidence-set volume, while maintaining the guarantee.

\begin{figure}[h]
	\centering
	\caption{Illustration of a Size--Weight Frontier\label{fig-frontier-illustration}}
	\includegraphics[width=0.58\linewidth]{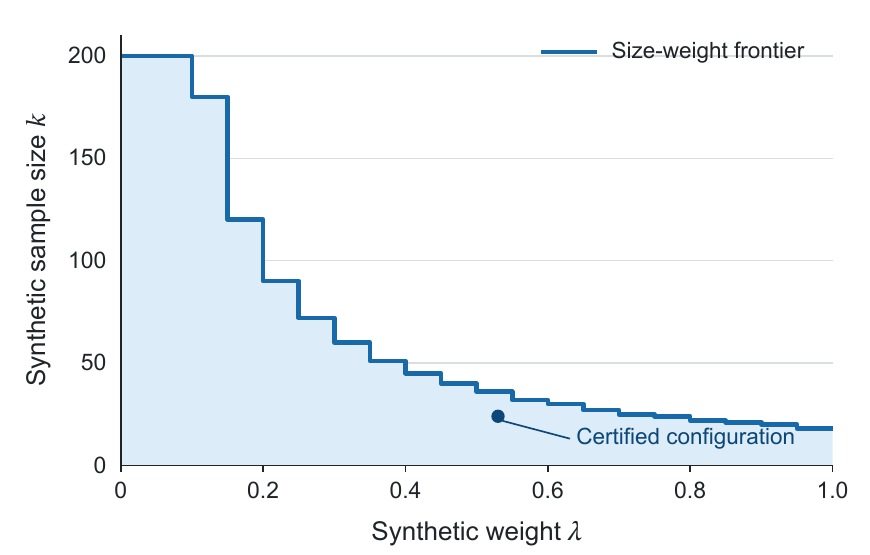}
	\bnotefig{Every size--weight configuration $(\weight,k)$ on or below the frontier leads to confidence sets with valid coverage.}
\end{figure}

\paragraph{Related work.}

Our work contributes to a growing literature on statistical inference with real and synthetic data \citep{SLS23,ABF23,BLL25,KSh25,BLL26,BGL26,AWL26}. Applications include imbalanced learning \citep{NXL24,MZh26}, online decision-making \citep{JPL25}, market research and surveys \citep{WZZ26,YXi26,YLT26,LWZ26}, among others. Our focus is to develop a data-driven framework that uses historical tasks to learn how much synthetic information can be incorporated while maintaining statistical validity. 

Closest to our work, \cite{HWW25} and \cite{TZr26} use historical tasks containing both real and synthetic data to calibrate statistical inference for a target task with only synthetic data. \cite{HWW25} adaptively selects the synthetic sample size, whereas \cite{TZr26} estimates the real--synthetic discrepancy for bias correction. By contrast, we study general data augmentation when the target task may contain real observations and calibrate a family of confidence sets indexed jointly by synthetic sample size and weight.

Combining weighted synthetic observations with real data also arises in Bayesian methods. Catalytic priors \citep{HSR20} use downweighted synthetic data from a simple predictive model to stabilize estimation of a more complex model. AI-powered Bayesian inference \citep{ORo25} constructs and tunes priors from a generative model's outputs. Our method instead targets frequentist coverage, and applies to general user-specified confidence set constructions.

\paragraph{Notation.}
For $n\in \ZZ_+$, let $[n]=\{1,\ldots,n\}$, with the convention $[0]=\varnothing$. For real numbers $a$ and $b$, write $a\wedge b=\min\{a,b\}$. We use $\Ber(\mu)$ to denote the Bernoulli distribution with mean parameter $\mu$. Throughout the paper, $\log$ denotes the natural logarithm.

\section{Problem Setup}\label{sec-problem}

\subsection{A Motivating Example}

Consider a collection of Yes/No survey questions. For a question $\query$, we are interested in estimating $\param_{\query}$, the proportion of the human population that would answer ``Yes''. We observe $n$ independent human responses $\dataset_n=\{\data_i\}_{i=1}^n\stackrel{\mathrm{i.i.d.}}{\sim}\Ber(\param_{\query})$. When the number of real samples $n$ is small, as is the case for new survey questions, estimating $\param_{\query}$ solely based on $\dataset_n$ has large statistical error, and the confidence interval can be wide.

A possible remedy is to augment the scarce real data with synthetic samples. Suppose we can prompt a large language model, with the question $\query$ and randomly generated respondent profiles, to produce $k$ independent synthetic responses $\syndataset_k=\{\syndata_i\}_{i=1}^k$. Then, we can make inference on $\param_{\query}$ using the combined dataset $\dataset_n \cup \syndataset_k$. As the synthetic data can be systematically biased, treating it na\"{i}vely as real samples may produce a narrow interval that misses the population mean $\param_{\query}$. We could maintain coverage validity by limiting the synthetic sample size \citep{SLS23,HWW25} or downweighting the synthetic samples \citep{HSR20,ORo25}. However, there is no general data-driven approach for choosing valid synthetic size--weight configurations. 

Our goal is to identify a frontier that specifies, for each candidate weight, how many synthetic samples can be augmented with the real samples while maintaining valid confidence interval coverage. 

\subsection{Weighted Synthetic Augmentation}

We now introduce the general problem framework, which extends the previous binary-survey example to a general population of tasks, target parameters, and confidence-set constructions. Our framework has two components. First, for each augmentation configuration, a user-specified procedure constructs a candidate confidence set. Then, we will develop methods that use historical tasks to determine which configurations achieve the target coverage for a future task. In this subsection, we focus on the first component.

Let $\{\dist_{\query}:\query\in\queryfam\}$ be a family of probability distributions on a sample space $\outputspace$, where $\query\in\queryfam$ indexes a statistical inference task. The target parameter for task $\query$ is
\[
\param_{\query}=\parammap(\dist_{\query}),
\]
for a specified distributional functional $\parammap$. The parameter may be a mean, a quantile, a vector, or a non-Euclidean object. For example, $\query$ may be a survey question ``What is your annual income?'', $\dist_{\query}$ the income distribution in the target population, and $\param_{\query}$ its $90$th percentile.

A target task $\query$ is drawn from a distribution $\querydist$. Conditional on $\query$, we observe $n$ independent real observations $\dataset_n=\{\data_i\}_{i=1}^n$ from $\dist_{\query}$. We also have access to a simulator that generates independent synthetic observations $\syndataset_k=\{\syndata_i\}_{i=1}^k$ from a synthetic distribution $\syndist_{\query}$. The real and
synthetic samples are conditionally independent given the task. The simulator may be misspecified: in general, $\syndist_{\query}\ne\dist_{\query}$ and $\parammap(\syndist_{\query})\ne\param_{\query}$.

We consider the following synthetic augmentation scheme. Fix a target miscoverage level $\alpha\in(0,1)$ and a maximum synthetic sample size $K\in\ZZ_+$. An \emph{augmentation configuration} is a pair $(\weight,k)$, where $k\in\{0,\ldots,K\}$ is the number of synthetic observations, and $\weight\in[0,\infty)$ is their weight relative to weight one for each real observation. Given $\dataset_n$, $\syndataset_k$ and $\weight$, a user-specified, possibly randomized procedure $\CI$ returns a set estimate of $\param_{\query}$,
\[
\CI(\dataset_n,\syndataset_k,\weight).
\]
This data-explicit notation emphasizes that the confidence set is a random function of the observed samples. We assess a configuration $(\weight, k)$ by its average coverage probability over the task population.

\begin{definition}[Coverage probability]
For $\weight\geq0$ and $k\in\{0,\ldots,K\}$, define
\begin{equation}
p_n(\weight,k)
=
\PP \Big(
\param_{\query}
\in
\CI(\dataset_n,\syndataset_k,\weight)
\Big).
\label{eq-coverage-true}
\end{equation}
Here the probability is taken over the randomness of $\query\sim\querydist$, and the real and synthetic samples $\dataset_n\sim\dist_{\query}$ and $\syndataset_k\sim\syndist_{\query}$ conditional on $\query$.
\end{definition}

Our goal is to identify a large, interpretable set of configurations $S\subseteq[0,\infty)\times\{0,\ldots,K\}$ such that
\begin{equation}
p_n(\weight,k)
\geq1-\alpha,
\qquad
\text{for every }(\weight,k)\in S.
\label{eq-valid-region}
\end{equation}

\begin{remark}[Marginal coverage]
We note that the coverage probability \eqref{eq-coverage-true} is marginal over the task distribution $\query\sim\querydist$. It is weaker than task-conditional coverage, which would require
\[
\PP 
\Big(
\param_{\query}
\in
\CI(\dataset_n,\syndataset_k,\weight)
~\Big|~
\query
\Big)
\geq1-\alpha
\]
for every task $\query$. For example, a coverage guarantee averaged across survey questions does not imply that the target coverage is achieved for every individual question.
\end{remark}

We place no structural assumption on $\CI$. The only requirement is that whenever there is no synthetic contribution (that is, when the synthetic sample size $k$ or weight $\weight$ is zero), it attains marginal coverage of at least $1-\alpha$.

\begin{assumption}[Real-only coverage]\label{assumption-fallback}
For every $\weight\geq0$ and $k\in\{0,\ldots,K\}$,
\[
p_n(\weight,0)\geq1-\alpha
\qquad\text{and}\qquad
p_n(0,k)\geq1-\alpha.
\]
\end{assumption}

The following two examples give instantiations of $\CI$ based on standard confidence-set procedures. The first considers confidence intervals for a scalar mean, while the second describes a bootstrap construction for a general parameter. We note that our calibration framework to be introduced later is not restricted to scalar or Euclidean parameters.

\begin{example}[Mean estimation]\label{example-mean}
Consider a univariate distribution $\dist_{\query}$ with finite second moment and target mean $\param_{\query}=\EE_{\data\sim\dist_{\query}}[\data]$. Let $\samplemean_n$ and $\samplesd_n$ denote the sample mean and sample standard deviation of $n\geq2$ real observations. Under standard central limit theorem (CLT) conditions, the interval
\begin{equation}
\left[
\samplemean_n-z_{1-\alpha/2}\frac{\samplesd_n}{\sqrt n},
\quad
\samplemean_n+z_{1-\alpha/2}\frac{\samplesd_n}{\sqrt n}
\right]
\label{eqn-CI-CLT-real}
\end{equation}
has asymptotic coverage $1-\alpha$, where $z_{1-\alpha/2}$ is the $(1-\alpha/2)$-quantile of the standard normal distribution. However, it may not satisfy Assumption~\ref{assumption-fallback} for finite $n$.

If $\dist_{\query}$ is supported on a known interval $[a,b]$, an empirical Bernstein confidence interval
provides an alternative with finite-sample guarantees \citep{MPo09}. It replaces the half-width $z_{1-\alpha/2} \samplesd_n / \sqrt n $ in \eqref{eqn-CI-CLT-real} with
\begin{equation}
\samplesd_n\sqrt{\frac{2\log(4/\alpha)}{n}}
+
\frac{7(b-a)\log(4/\alpha)}{3(n-1)}.
\label{eqn-CI-Bernstein-real}
\end{equation}

We next augment the real sample with $k\geq1$ synthetic observations. Let $\samplesynmean_k$ and $\samplesynsd_k$ denote their sample mean and sample standard deviation; set $\samplesynsd_1=0$ by convention. The weighted point estimator is
\begin{align}
\paramhat_{\query}(\weight,k)
&=
\frac{
\sum_{i=1}^n\data_i
+\weight\sum_{i=1}^k\syndata_i
}{n+\weight k}
=
\frac{n}{n+\weight k}\samplemean_n
+
\frac{\weight k}{n+\weight k}\samplesynmean_k.
\label{eq-weighted-mean}
\end{align}
Direct calculation yields
\[
\var \big[ \paramhat_{\query}(\weight,k) \big]
=
\frac{
n\var_{\data\sim\dist_{\query}}(\data)
+\weight^2k\var_{\syndata\sim\syndist_{\query}}(\syndata)
}{(n+\weight k)^2}.
\]
A plug-in estimate of this variance is
\begin{equation}
\samplevar_{\query}(\weight,k)
=
\frac{
n\samplesd_n^2
+\weight^2k(\samplesynsd_k)^2
}{(n+\weight k)^2}.
\label{eq-weighted-variance}
\end{equation}
Based on this, we use hybrid data to construct a CLT confidence interval
\[
\CI^{\mathsf C}(\dataset_n,\syndataset_k,\weight)
=
\left[
\paramhat_{\query}(\weight,k)
-z_{1-\alpha/2}\sqrt{\samplevar_{\query}(\weight,k)},
\quad
\paramhat_{\query}(\weight,k)
+z_{1-\alpha/2}\sqrt{\samplevar_{\query}(\weight,k)}
\right],
\]
and an empirical Bernstein interval
\begin{align}
\CI^{\mathsf B}(\dataset_n,\syndataset_k,\weight)
&=
\left[
\paramhat_{\query}(\weight,k)-r_{\query}(\weight,k),
\quad
\paramhat_{\query}(\weight,k)+r_{\query}(\weight,k)
\right]
\cap[a,b],
\label{eqn-CI-Bern-aug}
\end{align}
with
\begin{align}
	r_{\query}(\weight,k)
	&=
	\sqrt{
		2\samplevar_{\query}(\weight,k)\log(4/\alpha)
	}
	+
	\frac{
		7(b-a) \log(4/\alpha)
	}{
		3 (n + \weight k - 1)
	}.
	\label{eqn-margin-Bern-aug}
\end{align}
When $\weight=0$, they reduce to the corresponding real-data-only confidence intervals.
\end{example}

\begin{example}[Bootstrap confidence sets for general parameters]\label{example-bootstrap}
For a general functional $\param_{\query}=\parammap(\dist_{\query})$, define the
weighted empirical distribution
\[
\widehat P_{\query,\weight,k}
=
\frac{1}{n+\weight k}
\left(
\sum_{i=1}^n\delta_{\data_i}
+\weight\sum_{i=1}^k\delta_{\syndata_i}
\right)
\]
and the point estimate
$\paramhat_{\query}(\weight,k)=\parammap(\widehat
P_{\query,\weight,k})$. To quantify uncertainty, we generate $B$ bootstrap replicates. In replicate $b=1,\ldots,B$, draw $\{\databoot_i\}_{i=1}^n$ with replacement from $\dataset_n$ and $\{\syndataboot_i\}_{i=1}^k$ with replacement from $\syndataset_k$. Compute
\[
\paramboot_{\query}(\weight,k)
=
\parammap\!\left[
\frac{1}{n+\weight k}
\left(
\sum_{i=1}^n\delta_{\databoot_i}
+\weight\sum_{i=1}^k\delta_{\syndataboot_i}
\right)
\right].
\]
The bootstrap replicates can be converted into a set estimate
$\CI(\dataset_n,\syndataset_k,\weight)$ using a percentile, basic, studentized, or any other problem-specific construction.
\end{example}

\subsection{The Oracle Size--Weight Frontier}

The preceding examples show how to construct a confidence set for a given configuration $(\weight,k)$. We now define the population benchmark that identifies the $(\weight,k)$ configurations attaining the target coverage level. 

Choose a finite ordered collection of candidate weights, independently of the calibration data:
\begin{equation}
\Lambda
=
\{\weight_m\}_{m=0}^{M},
\qquad
0=\weight_0<\weight_1<\cdots<\weight_M<\infty.
\label{eq-weight-grid}
\end{equation}
We will work with configurations whose weights belong to this prespecified collection $\Lambda$. At zero weight $\weight_0=0$, synthetic data do not contribute to the confidence interval. By Assumption \ref{assumption-fallback}, $p_n(\weight_0,k)\ge 1-\alpha$ for all $k\le K$, and we set $k^*(\weight_0)=K$. For $m\in[M]$, we recursively define
\begin{equation}
k^*(\weight_m)
=
\max\left\{
k\in\{0,\ldots,k^*(\weight_{m-1})\}:
p_n(\weight_m,i)\geq1-\alpha
\ \text{for every }i\in[k]
\right\}.
\label{eqn-oracle}
\end{equation}
That is, $k^*(\weight_m)$ is the largest synthetic sample size $k$ at weight $\weight_m$ for which every sample size up to $k$ attains the target $(1-\alpha)$ coverage. Moreover, it is constrained not to exceed its counterpart $k^*(\weight_{m-1})$ at the previous weight $\weight_{m-1}$. We note that $k^*(\weight_m)$ is well-defined because $k=0$ is always feasible under the convention $[0]=\varnothing$. This leads to the \emph{oracle size--weight frontier}
\[
\left\{ \big(\weight,k^*(\weight) \big):\weight\in\Lambda\right\}.
\] 

Two conservative restrictions are built into this definition: (i) at a fixed weight, the frontier includes a sample size only if every smaller positive sample size is also valid; (ii) across the ordered weights, the recursion forces $k^*(\weight_m)\leq k^*(\weight_{m-1})$. Consequently, every $(\weight,k)$ configuration that lies on or below the frontier attains the target coverage:
\[
p_n(\weight,k)\geq1-\alpha,
\qquad
\forall \weight\in\Lambda,
\quad
0\leq k\leq k^*(\weight).
\]
As the task-marginal probabilities $p_n(\weight,k)$ are unknown, the oracle frontier \eqref{eqn-oracle} is not directly available. We will now develop data-driven approaches that use historical tasks to learn this frontier, and quantify the coverage error of the learned frontier.

\section{Learning the Size--Weight Frontier}\label{sec-method}

Suppose we observe $J$ historical tasks
$\query_1,\ldots,\query_J$ drawn independently from $\querydist$. For task $\query_j$, we observe $n_j$ real samples
\[
\dataset_j
=
\{\data_{j,i}\}_{i=1}^{n_j},
\qquad
\data_{j,i}\stackrel{\mathrm{i.i.d.}}{\sim}\dist_{\query_j},
\]
and generate $K$ synthetic samples
\[
\syndataset_j
=
\{\syndata_{j,i}\}_{i=1}^{K},
\qquad
\syndata_{j,i}\stackrel{\mathrm{i.i.d.}}{\sim}\syndist_{\query_j}.
\]
We refer to $\dataset_{\mathrm{cal}}
=
\{(\query_j,\dataset_j,\syndataset_j)\}_{j=1}^J$ as the \emph{calibration data}, which we will use to learn the frontier. We make the following assumption.

\begin{assumption}[Independent calibration data]\label{assumption-independence}
The sample sizes $n_1,\ldots,n_J$ are deterministic and satisfy $\min_{j\in[J]}n_j\geq n$. Conditional on the tasks $\query_1,\ldots,\query_J$, the two collections $\{\dataset_j\}_{j=1}^J$ and $\{\syndataset_j\}_{j=1}^J$ are independent, and the pairs $(\dataset_j,\syndataset_j)$ are independent across $j\in[J]$.
\end{assumption}

For each historical task $\query_j$, we construct a real dataset $\dataset_{j,n}=\{\data_{j,i}\}_{i=1}^n$ of size $n$, and a synthetic dataset $\syndataset_{j,k}=\{\syndata_{j,i}\}_{i=1}^k$ of size $k$. Define a set estimate
\begin{equation}
\CI_{j,n}(\weight,k)
=
\CI(\dataset_{j,n},\syndataset_{j,k},\weight).
\label{eq-historical-candidate}
\end{equation}
Write
$\param_j=\parammap(\dist_{\query_j})$ for the unknown target parameter on task
$j$. Intuitively, we would like to choose configurations $(\weight, k)$ such that the empirical coverage on the $J$ historical tasks achieves the nominal level:
\begin{align}
\frac{1}{J} \sum_{j=1}^J \one \Big(
\param_j \in \CI_{j,n}(\weight,k)
\Big) \geq 1 - \alpha.
\end{align}
However, as the true parameter $\param_j$ is unknown, one cannot directly check whether $\param_j \in \CI_{j,n}(\weight,k)$ holds. To resolve this issue, we use the real data $\dataset_j$ to construct a reference confidence set $\CI_j$ for $\param_j$ as a (noisy) proxy, and then decide based on the empirical frequency of inclusion $\CI_j \subseteq \CI_{j,n}(\weight,k)$. Below we introduce two constructions, both of which use a constant parameter $\gamma\in(0,1)$.

\paragraph{Full-sample proxy.}
Using the full real dataset $\dataset_j$, construct a reference set $\CI_j^{\full}$ satisfying
\begin{equation}
\PP \Big( \param_j\in\CI_j^{\full} \Big)
\geq1-\gamma\alpha.
\label{eq-full-reference-coverage}
\end{equation}
The probability in \eqref{eq-full-reference-coverage} is taken over the task $\query_j$ and its real data $\dataset_j$. It is sufficient but not necessary to have the task-conditional guarantee
\[
\PP \Big( \param_j\in\CI_j^{\full}
~\Big|~
\query_j \Big)
\geq1-\gamma\alpha.
\]
We define the \emph{full-sample coverage proxy}
\begin{equation}
\widehat p_n^{\full}(\weight,k)
=
\frac1J\sum_{j=1}^J
\one \Big(
\CI_j^{\full}\subseteq\CI_{j,n}(\weight,k)
\Big).
\label{eq-proxy-full}
\end{equation}

\paragraph{Split-sample proxy.}
Alternatively, when $n_j>n$ for all $j$, we can use only the held-out real observations $\dataset_j\setminus\dataset_{j,n}$ to construct a reference set $\CI_j^{\splt}$ satisfying the task-conditional guarantee
\begin{equation}
\PP \Big( \param_j\in\CI_j^{\splt}
~\Big|~
\query_j
\Big)
\geq1-\gamma.
\label{eq-split-reference-coverage}
\end{equation}
Due to data splitting, this reference set is conditionally independent of $\CI_{j,n}(\weight,k)$ given $\query_j$. We define the \emph{split-sample coverage proxy}
\begin{equation}
\widehat p_n^{\splt}(\weight,k)
=
\frac1J\sum_{j=1}^J
\one \Big( 
\CI_j^{\splt}\subseteq\CI_{j,n}(\weight,k)
\Big).
\label{eq-proxy-split}
\end{equation}

The following lemma relates both proxies to the true coverage probability in \eqref{eq-coverage-true}.

\begin{lemma}[From set containment to coverage]\label{lem-proxies}
For every $(\weight,k)\in[0,\infty)\times\{0,\ldots,K\}$,
\[
p_n(\weight,k)
\geq
\EE\widehat p_n^{\full}(\weight,k)-\gamma\alpha
\qquad\text{and}\qquad
p_n(\weight,k)
\geq
\frac{
\EE\widehat p_n^{\splt}(\weight,k)-\gamma
}{1-\gamma}.
\]
\end{lemma}

\begin{proof}[Proof of \Cref{lem-proxies}]
See \Cref{sec-proofs-proxies}.
\end{proof}

Let $\widehat p_n$ denote either $\widehat p_n^{\full}$ or $\widehat p_n^{\splt}$. According to \Cref{lem-proxies}, $\EE\widehat p_n(\weight,k)\geq 1-(1-\gamma)\alpha$ implies $p_n(\weight,k)\geq1-\alpha$. Thus, we will use $1-(1-\gamma)\alpha$ as the acceptance threshold for $\widehat p_n$, which is slightly higher than the nominal level.

Based on the observation above, we propose to estimate the frontier $k^*$ by replacing the unknown coverage probabilities in \eqref{eqn-oracle} by their empirical proxies. Set $\widehat k(\weight_0)=K$. For $m\in[M]$, we recursively define
\begin{align}
\widehat k(\weight_m)
=
\max\Big\{
&k\in\{0,\ldots,\widehat k(\weight_{m-1})\}:~
\widehat p_n(\weight_m,i)\geq
1-(1-\gamma)\alpha
\ \text{for every }i\in[k]
\Big\}.
\label{eq-frontier-est}
\end{align}
The set in \eqref{eq-frontier-est} is nonempty because $k=0$ is always feasible. The recursive upper bound makes $\widehat k$ non-increasing across the ordered weights. \Cref{alg-frontier-learning} summarizes the procedure.

\begin{algorithm}[t]
\caption{Learning the size--weight frontier}\label{alg-frontier-learning}
\begin{algorithmic}[1]
\REQUIRE Calibration data $\dataset_{\mathrm{cal}}$, set estimator $\CI$, grid of weights $\Lambda$, simulation budget $K$, nominal miscoverage probability $\alpha \in (0,1)$, proxy type
$q\in\{\full,\splt\}$, constant $\gamma \in (0,1)$.
\STATE Set $\tau=1-(1-\gamma)\alpha$.
\FOR{$j=1,\ldots,J$}
    \STATE Construct the reference set $\CI_j^q$ and the candidate sets
    $\CI_{j,n}(\weight_m,k)$ for $m\in[M]$ and $k\in[K]$.
\ENDFOR
\STATE Compute $\widehat p_n^q(\weight_m,k)$ using
\eqref{eq-proxy-full} or \eqref{eq-proxy-split}.
\STATE Set $\widehat k(\weight_0)=K$.
\FOR{$m=1,\ldots,M$}
    \STATE Compute $\widehat k(\weight_m)$ from \eqref{eq-frontier-est}.
\ENDFOR
\RETURN $\widehat k:\Lambda\to\{0,\ldots,K\}$
\end{algorithmic}
\end{algorithm}

\begin{remark}[Comparison of proxies]
The full-sample proxy uses all available real observations and requires only marginal coverage validity of its reference sets. However, as its reference level is $1-\gamma\alpha$, the reference set $\CI_j^{\full}$ can be wide, making containment $\CI_j^{\full}\subseteq\CI_{j,n}(\weight,k)$ harder to achieve. The split-sample proxy uses the lower reference level $1-\gamma$ and can reuse the same reference set $\CI_j^{\splt}$ when the target level $\alpha$ changes. The cost is that it reserves data for the reference set and requires task-conditional validity.
\end{remark}

\begin{remark}[The case of no real data]
When $n=0$, many confidence-set constructions are invariant to the value of the synthetic weight $\weight> 0$, because all synthetic samples receive the same weight. In that case, we only need to select the synthetic sample size $k$, and the problem reduces to the one studied by \citet{HWW25}. Their general algorithm coincides with \Cref{alg-frontier-learning} using the split-sample proxy.
\end{remark}

We next establish a simultaneous coverage guarantee for all $(\weight,k)$ configurations on or below the learned frontier $\widehat{k}$.

\begin{theorem}[Uniform coverage]
\label{thm-uniform-coverage}
Fix $n\in\NN$, $K\in\ZZ_+$, $\delta\in(0,1)$, and the finite collection
$\Lambda$ in \eqref{eq-weight-grid}. Let
Assumptions~\ref{assumption-fallback} and \ref{assumption-independence} hold, and let $\widehat k$ be the output of \Cref{alg-frontier-learning}. Set $c=1$ for the full-sample proxy and
$c=(1-\gamma)^{-1}$ for the split-sample proxy. With probability at least
$1-\delta$ over $\dataset_{\mathrm{cal}}$, we have
\begin{equation}
\min_{\substack{
\weight\in\Lambda\\
0\leq k\leq\widehat k(\weight)
}}
p_n(\weight,k)
\geq
1-\alpha
-c\sqrt{
\frac{\log\{(K\wedge M)/\delta\}}{2J}
}.
\label{eq-uniform-coverage-bound}
\end{equation}
\end{theorem}

\begin{proof}[Proof of \Cref{thm-uniform-coverage}]
See \Cref{sec-thm-uniform-coverage-proof}.
\end{proof}

\Cref{thm-uniform-coverage} states that with high probability, all configurations on or below the learned frontier $\widehat{k}$ attain the target coverage level $(1-\alpha)$, up to an error of order $J^{-1/2}$. The logarithmic factor in the error term depends on $K\wedge M$. Thus, once the finite grid has at least $K+1$ weights (that is, $M\ge K$), refining it further does not increase the statistical penalty. 

Thanks to the uniform coverage guarantee, any configuration subsequently selected on or below the learned frontier retains coverage validity. In particular, the same calibration data can be used first to learn the frontier and then to select a configuration within the learned region, without incurring an additional coverage penalty. For example, one may choose a configuration that minimizes the average confidence-set volume across the calibration tasks. The following result gives a coverage bound for any selected configuration, where the coverage probability averages over both the calibration data and an independent future task.

\begin{corollary}[Calibration-data-dependent selection]
\label{cor-marginal-coverage}
Consider the setting of \Cref{thm-uniform-coverage} with
$K\wedge M\geq2$. Let
$(\widetilde\weight,\widetilde k)$ be any possibly randomized function of the
calibration data satisfying
\[
\widetilde\weight\in\Lambda,
\qquad
0\leq\widetilde k\leq\widehat k(\widetilde\weight).
\]
For an independent future task and its data,
\[
\PP \Big( 
\param_{\query}
\in
\CI(\dataset_n,\syndataset_{\widetilde k},\widetilde\weight)
\Big)
\geq
1-\alpha
-c\sqrt{
\frac{2\log(K\wedge M)}{J}
}.
\]
\end{corollary}

\begin{proof}[Proof of \Cref{cor-marginal-coverage}]
See \Cref{sec-cor-marginal-coverage-proof}.
\end{proof}

\section{Numerical Experiments}\label{sec-experiments}

In this section, we evaluate the proposed synthetic augmentation methods on a large-scale opinion survey dataset, using LLMs as synthetic simulators. We examine three complementary aspects: (i) coverage validity of the synthetic-augmented confidence intervals; (ii) efficiency gains from synthetic augmentation, measured by the reduction in confidence interval width relative to using real data alone; and (iii) effective synthetic sample size $\weight\cdot\widehat{k}(\weight)$, which quantifies how much synthetic information can be incorporated while maintaining valid inference. The code for implementing the methods and reproducing the experiments is available at \url{https://github.com/ch3702/synthetic-augmentation-frontier}.

\subsection{Experiment Setup}\label{sec-experiments-setup}

\paragraph{Dataset.} We use the \textsc{WorldValuesBench} dataset \citep{ZMT24}, curated from the World Values Survey \citep{WVS20}. The dataset consists of survey questions asked across $64$ countries, covering $12$ themes including social values, security, migration, and so on. We follow the same preprocessing procedure in \cite{ILW25}, which yields $235$ questions. Each question admits a numerical response mapped to $[-1,1]$, and the target parameter is the mean response: $\parammap(\dist)=\EE_{\data\sim\dist}[\data]$. The questions contain responses from a large population of $96{,}200$ individuals. From this population, we draw small real-data samples to simulate the limited-real-data regime. The large population also allows us to compute almost exact population means, which we use as the ground truth when evaluating coverage on held-out questions.

\paragraph{LLMs.} Following the preprocessing and synthetic response generation protocol of \cite{ILW25}, we evaluate two LLMs as synthetic simulators: GPT-4o \citep{GPT4o} and GPT-5 mini \citep{GPT5}. For each LLM and each question, we allow at most $K=200$ synthetic responses to be incorporated into a confidence interval.

\paragraph{Confidence interval construction.} We compute the augmented confidence interval $\CI (\dataset_n,\syndataset_k,\weight)$ using the empirical Bernstein interval given in \eqref{eqn-CI-Bern-aug} and \eqref{eqn-margin-Bern-aug} with range $[a,b]=[-1,1]$. The corresponding real-data confidence interval has the finite-sample coverage property in Assumption \ref{assumption-fallback}.

\paragraph{Proxy construction and hyperparameters.}
We set the target coverage level to $1-\alpha=0.9$. For each question, we take $N=500$ real samples; that is, $n_j=N$ for all $j$. The first $n$ responses are used to construct the augmented confidence interval. For the full-sample proxy, all $N$ responses are used to construct $\CI_j^{\full}$, so the proxy data overlap with the $n$ responses used for augmentation. For the split-sample proxy, the remaining $(N-n)$ responses are held out to construct $\CI_j^{\splt}$ and are therefore independent of the $n$ responses used for augmentation. We take $\gamma=0.5$ for both full-sample and split-sample proxies.

\subsection{Experiment Procedure}\label{sec-experiments-procedure}

We describe the data split and evaluation metrics used to assess the learned frontier.

\paragraph{Calibration and test data.}
Denote the collection of questions by $\metadataset=\{(\query_j,\dataset_j)\}_{j\in\cJ}$. We randomly partition $\cJ$ into a calibration set $\cJ_{\ca}$ and a test set $\cJ_{\te}$ with $|\cJ_{\ca}|:|\cJ_{\te}|=3:2$. For each partition, we also draw a fresh real-response pool for every question. We repeat this procedure over $50$ independent random partitions and response draws, and average all reported metrics across repetitions.

\paragraph{Learning the empirical frontier.}
On the calibration set $\{(\query_j,\dataset_j)\}_{j\in\cJ_{\ca}}$, we apply \Cref{alg-frontier-learning} with either the full-sample or split-sample coverage proxy to learn the empirical frontier. We consider real sample sizes $n\in\{10,30,50\}$; the learned frontier associated with real sample size $n$ is denoted by $\weight\mapsto\widehat{k}_n(\weight)$. The synthetic weight is discretized on the grid $\{\weight_m\}_{m=0}^M = \{0.05r:r=0,1,\ldots,20\}$, and the synthetic sample size ranges over $k\in\{1,\ldots,K\}$ with $K=200$.

\paragraph{Evaluation metrics.}
On the test set $\{(\query_j,\dataset_j)\}_{j\in\cJ_{\te}}$, for each $j\in\cJ_{\te}$, the learned frontier $\weight\mapsto\widehat{k}_n(\weight)$ leads to the synthetic-augmented confidence interval $\CI_{j,n}\big(\weight,\widehat{k}_n(\weight)\big)=\CI\big(\dataset_{j,n},\syndataset_{j,\widehat{k}_n(\weight)},\weight\big)$. We evaluate the learned frontier using the following three metrics.

\begin{enumerate}[label=(\roman*)]
\item \emph{Coverage probability.} Let $\bar{\param}_j$ denote the mean computed from the full pool of real responses for question $j$. As each test question has approximately $100,000$ real responses, we treat $\bar{\param}_j$ as the ground truth. For a configuration $(\weight,k)$, we estimate its coverage over the test set by
\[
\widehat{p}_{n,\te}(\weight,k) =
\frac{1}{|\cJ_{\te}|}
\sum_{j\in\cJ_{\te}}
\one \left(
\bar{\param}_j\in
\CI_{j,n}(\weight,k)
\right).
\]
We report the empirical coverage for the learned frontier, $\widehat{p}_{n,\te}(\weight,\widehat{k}_n(\weight))$, and compare it with the nominal coverage level $1-\alpha=0.9$.

\item \emph{Confidence interval width reduction.} We report the average relative reduction in confidence interval width on the learned frontier compared with the real-data interval,
\[
\overline{R}_n(\weight)
=
\frac{1}{|\cJ_{\te}|}
\sum_{j\in\cJ_{\te}}
\left(
1-
\frac{\Big|\CI_{j,n}\big(\weight,\widehat{k}_n(\weight)\big)\Big|}
{\left|\CI_{j,n}(0,0)\right|}
\right).
\]
Thus, $\overline{R}_n(0)=0$. For $\weight>0$, larger positive values indicate greater efficiency gains from synthetic augmentation.

\item \emph{Effective synthetic sample size.} We report $\weight\cdot\widehat{k}_n(\weight)$, which measures the total amount of synthetic information incorporated into the confidence interval. It accounts jointly for the number of synthetic samples $\widehat{k}_n(\weight)$ and their individual weight $\weight$.
\end{enumerate}

\subsection{Experiment Results}\label{sec-experiments-results}

We present results for our methods using both the full proxy and split proxy.

\paragraph{Learned versus oracle frontiers.} Before turning to the evaluation metrics, we visualize the learned frontier itself. We compare the learned frontier $\weight\mapsto\widehat{k}_n(\weight)$ against the oracle frontier $\weight\mapsto k^*_{n,\te}(\weight)$ for the test set, defined by $k^*_{n,\te}(0)=K$, and for each $m\in[M]$,
\[
k^*_{n,\te}(\weight_m) = \max\left\{
k \in \{0,\ldots, k^*_{n,\te}(\weight_{m-1})\}: ~
\widehat{p}_{n,\te}(\weight_m,i) \ge 1-\alpha
\ \text{for every }i\in[k]
\right\}.
\] 
The oracle frontier represents the largest synthetic sample size such that all smaller sample sizes also attain target coverage across the test questions.

\Cref{fig-frontier} plots the learned frontier $\weight\mapsto\widehat{k}_n(\weight)$ against $\weight\mapsto k_{n,\te}^*(\weight)$ for GPT-4o, for a single calibration--test split and $n\in\{10,30,50\}$. We see that for small weights $\weight$ (around $0.2$ or smaller), both frontiers are capped at the simulation budget $K=200$: even when all $K$ synthetic samples are incorporated, the augmented interval still satisfies the coverage constraint, so the frontier is limited by the budget rather than by validity. 

For larger weights $\weight$, the coverage constraint becomes binding and both frontiers decay steadily as $\weight$ approaches $1$. Across all three sample sizes and both proxy constructions, the learned frontier never exceeds the oracle while staying close to it. Hence \Cref{alg-frontier-learning} recovers a faithful yet slightly conservative estimate of the oracle: it admits nearly as much synthetic data as the oracle frontier without ever overshooting it, which is consistent with the coverage guarantee in \Cref{thm-uniform-coverage}. Moreover, the split-sample proxy yields a frontier marginally closer to the oracle than the full-sample proxy.

\begin{figure}[h]
    \centering
    \caption{Learned and Oracle Size--Weight Frontiers for GPT-4o\label{fig-frontier}}
    {\small (a) Full-sample proxy}\\
    \includegraphics[width=0.95\textwidth]{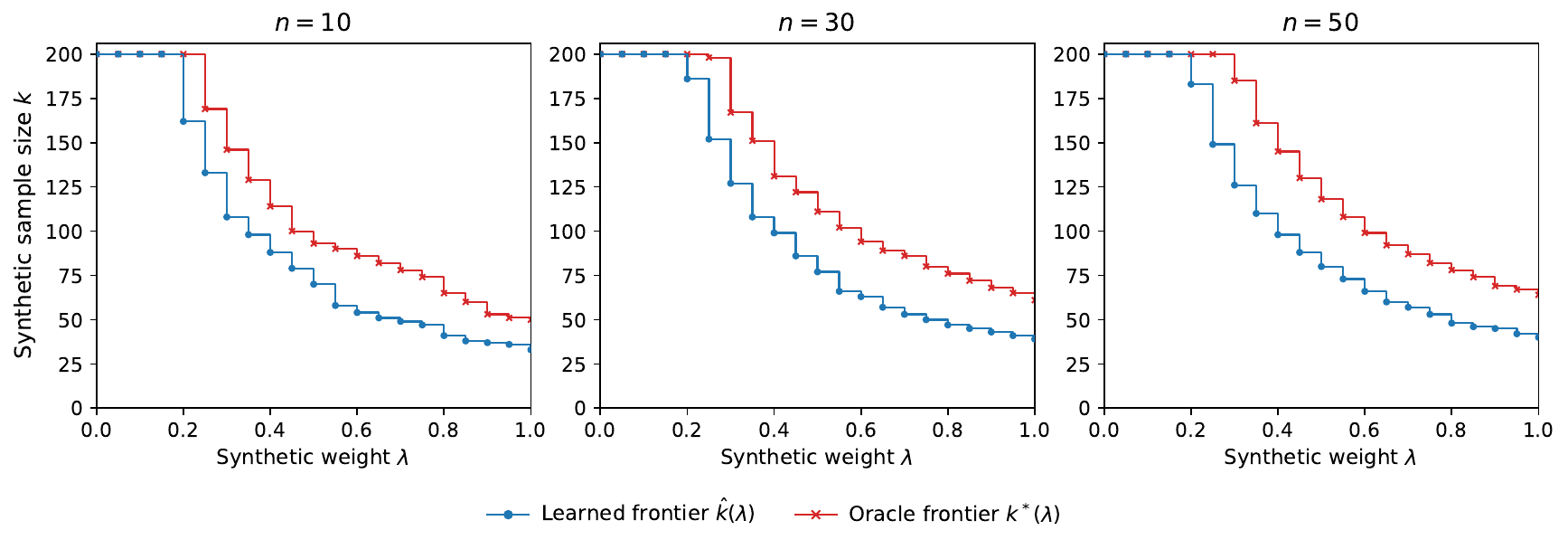}\\
    {\small (b) Split-sample proxy}\\
    \includegraphics[width=0.95\textwidth]{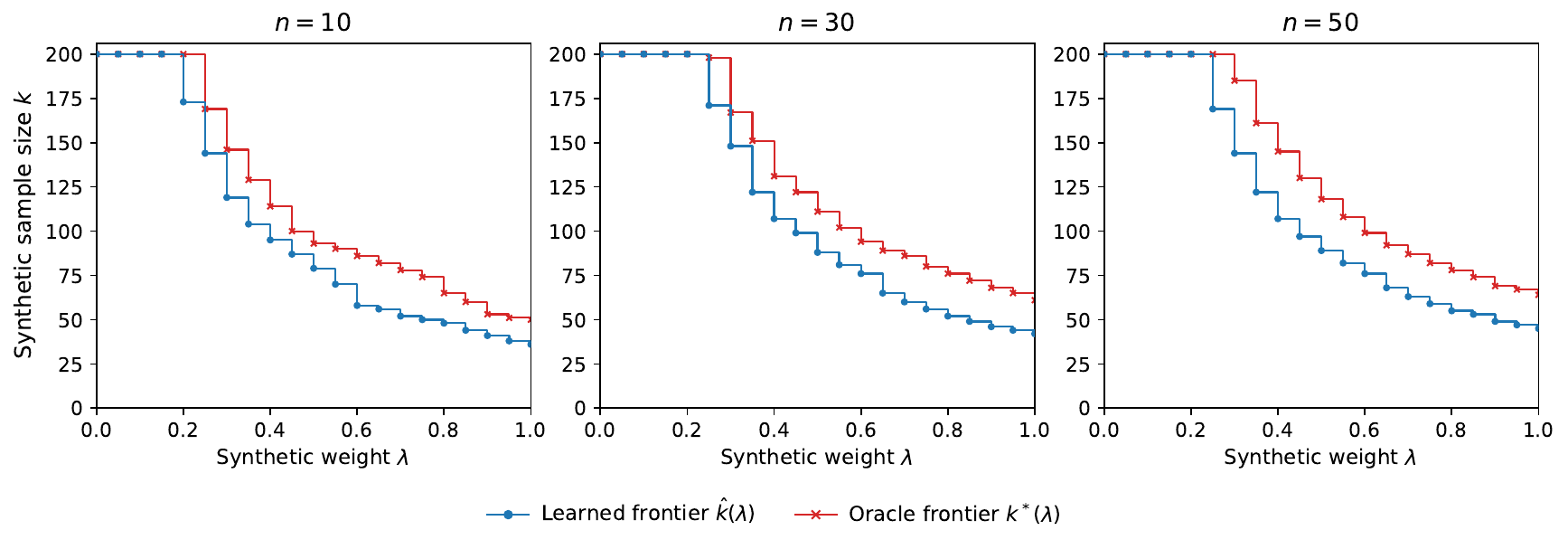}
    \bnotefig{This figure shows the learned frontier $\weight\mapsto\widehat{k}_n(\weight)$ (blue circles) and the oracle frontier $\weight\mapsto k_{n,\te}^*(\weight)$ (red crosses) for GPT-4o at $1-\alpha=0.9$, for a single calibration--test split and $n\in\{10,30,50\}$. Panel (a) uses the full-sample proxy and panel (b) uses the split-sample proxy.}
\end{figure}

\paragraph{Coverage validity.} \Cref{fig-coverage} plots the empirical test coverage $\widehat{p}_{n,\te}(\weight,\widehat{k}_n(\weight))$ at the learned frontier. The gray region indicates the small-$\weight$ regime in which the learned frontier is constrained by the maximum simulation budget $K=200$. Across both LLMs, both proxy constructions, and all real sample sizes $n$, coverage remains above the nominal level $1-\alpha=0.9$ for every evaluated $\weight$. These results are consistent with the uniform coverage guarantee in \Cref{thm-uniform-coverage}. 

\begin{figure}[H]
    \centering
    \caption{Empirical Test Coverage $\widehat{p}_{n,\te}(\weight,\widehat{k}_n(\weight))$ at the Learned Frontier \label{fig-coverage}}
    \includegraphics[width=\textwidth]{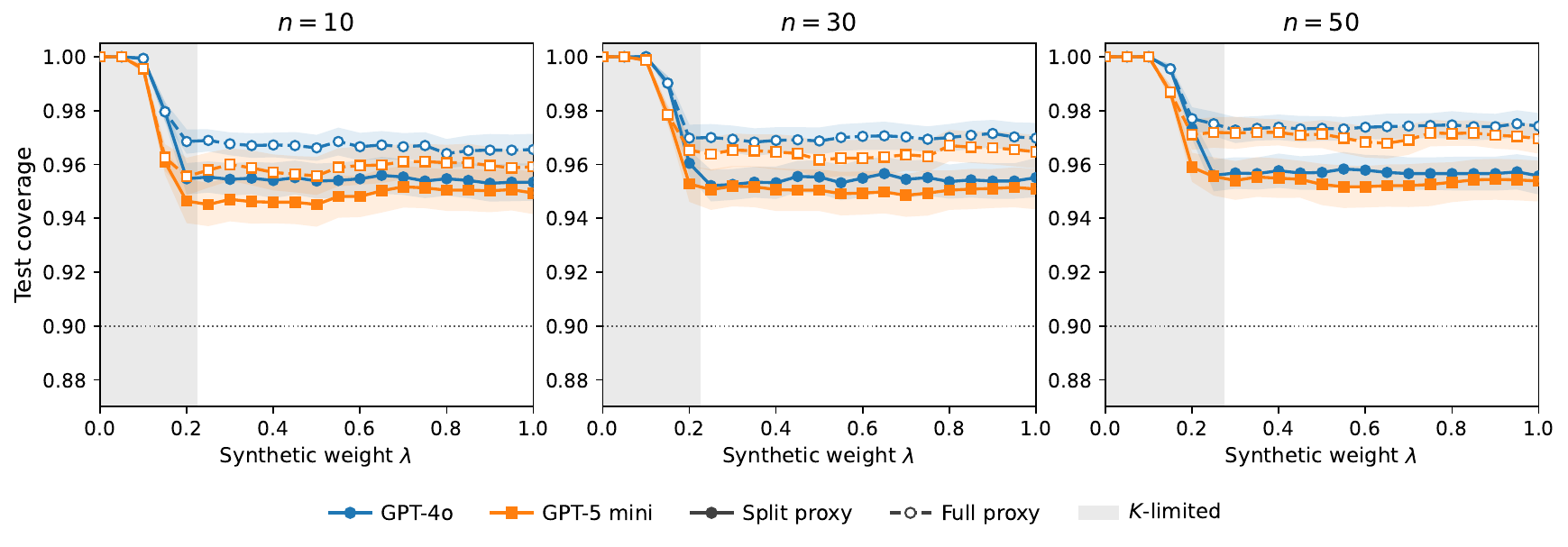}
    \bnotefig{The horizontal dotted line marks the nominal coverage level $1-\alpha=0.9$. Solid lines denote the split-sample proxy and dashed lines denote the full-sample proxy. Error bands have half-width equal to $1.96$ multiplied by the standard errors over $50$ repetitions, and the gray region denotes the regime constrained by the simulation budget $K=200$.}
\end{figure}

\paragraph{Confidence interval width reduction.} \Cref{fig-ciwidth} reports the average relative width reduction $\overline{R}_n(\weight)$. Again, the gray region indicates the small-$\weight$ regime in which the learned frontier is constrained by the maximum simulation budget $K=200$. We observe that synthetic augmentation produces substantial efficiency gains under both proxy constructions. For $n=10$ and $n=30$, the interval width is reduced by roughly $42\%$--$53\%$ once $\weight$ moves beyond the budget-constrained region; for $n=50$, the reduction remains approximately $34\%$--$44\%$. Across both LLMs and all reported values of $n$, the split-sample proxy yields reductions slightly larger than those obtained by the full-sample proxy. The relative improvement becomes smaller at $n=50$, as the real-data-only confidence interval is already narrower when more real observations are available.

\begin{figure}[h]
    \centering
    \caption{Average Relative Confidence Interval Width Reduction $\overline{R}_n(\weight)$\label{fig-ciwidth}}
    \includegraphics[width=\textwidth]{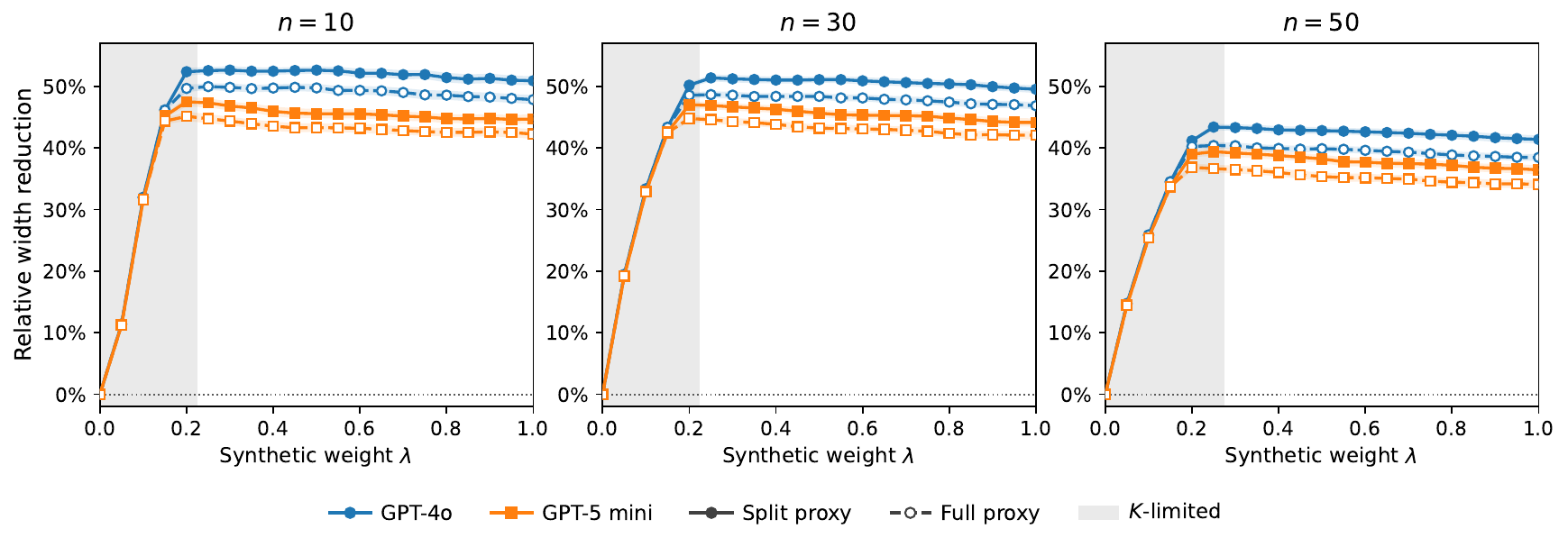}
    \bnotefig{Solid lines denote the split-sample proxy and dashed lines denote the full-sample proxy. Error bands have half-width equal to $1.96$ multiplied by the standard errors over $50$ repetitions, and the gray region denotes the regime constrained by the simulation budget $K=200$.}
\end{figure}

\paragraph{Effective synthetic sample size.} \Cref{fig-n-eff} reports the effective synthetic sample size $\weight\cdot\widehat{k}_n(\weight)$. This quantity reflects the amount of synthetic information that is augmented while preserving statistical validity. For small weights $\weight$, the learned frontier reaches the simulation budget $K$, so the effective size grows as $K\weight$. Away from this regime, e.g., for $\weight \ge 0.25$, the effective sample size $\weight\cdot\widehat{k}_n(\weight)$ is stable across different values of $\weight$. Moreover, the split-sample proxy consistently admits more synthetic information than the full-sample proxy, for both LLMs.

The effective synthetic sample size generally increases with $n$, even though the relative width reduction decreases at larger $n$. These observations are compatible: a larger real sample size $n$ allows the learned frontier to admit more synthetic observations, while the same amount of synthetic information produces smaller improvement over an increasingly precise real-data-only interval. The effective sample size also provides an interpretable measure of simulator fidelity. Under this lens, GPT-4o exhibits greater alignment with the real population than GPT-5 mini on the \textsc{WorldValuesBench} dataset, and the split-sample proxy extracts more synthetic information from both LLMs.

\begin{figure}[h]
    \centering
    \caption{Effective Synthetic Sample Size $\weight\widehat{k}_n(\weight)$ at the Learned Frontier\label{fig-n-eff}}
    \includegraphics[width=\textwidth]{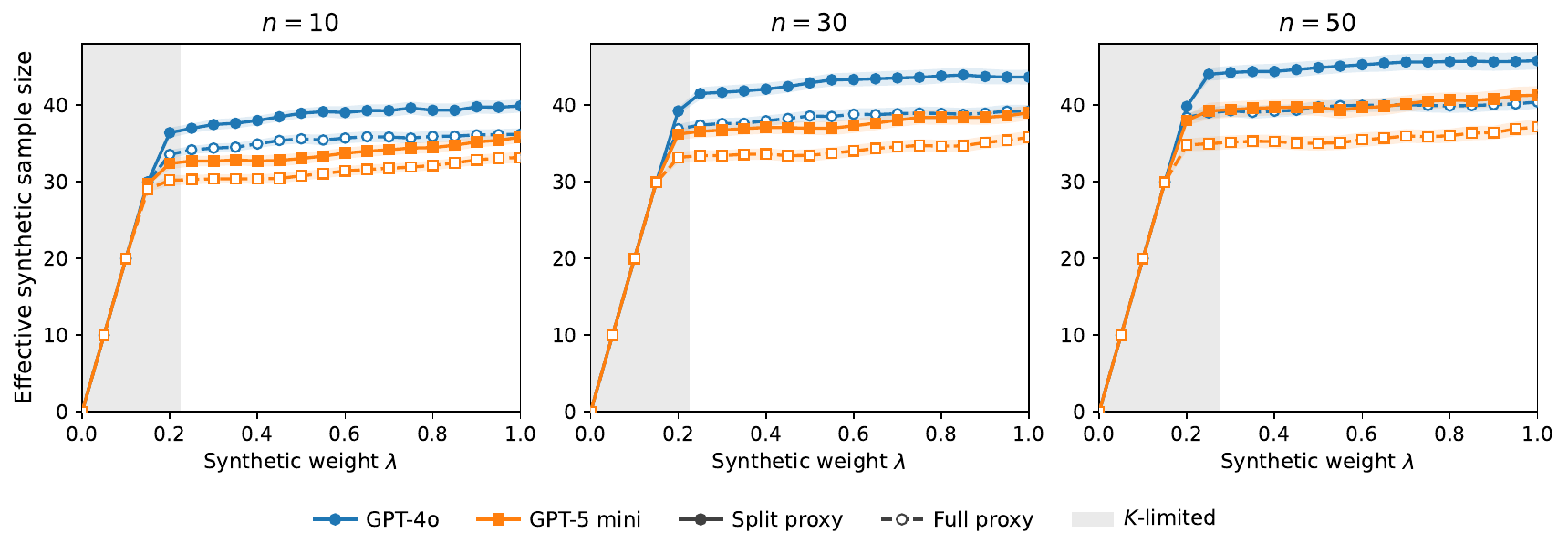}
    \bnotefig{Solid lines denote the split-sample proxy and dashed lines denote the full-sample proxy. Error bands have half-width equal to $1.96$ multiplied by the standard errors over $50$ repetitions, and the gray region denotes the regime constrained by the simulation budget $K=200$.}
\end{figure}

\section{Discussion}\label{sec-discussion}

We introduced a general framework for calibrating both the number of augmented synthetic samples and their weight when constructing confidence sets from scarce real data. The estimated size--weight frontier provides a simultaneous coverage bound for all configurations on or below it, allowing a final configuration to be selected according to a downstream objective based on the calibration data. The theoretical guarantees are marginal over the task population, and rely on the historical and future tasks being independent draws from the same distribution, a prespecified finite weight grid, and comparable tasks and data-generation processes at calibration and deployment. A promising future direction is to develop shift-aware or covariate-conditional frontiers that adapt synthetic augmentation to observable task characteristics and remain reliable under distribution shifts.

\section*{Acknowledgement}
The research is partially supported by the National Science Foundation grants DMS-2515679 and DMS-2544147. The authors used ChatGPT 5.6 Sol Extra High for language editing and Claude Opus 4.7 and 4.8 to assist with code implementation.

\newpage

\appendix
\crefalias{section}{appendix}
\crefalias{subsection}{appendix}

\section{Proofs}

\subsection{Proof of \Cref{lem-proxies}}\label{sec-proofs-proxies}

Fix a configuration $(\weight,k)$, and suppress it from the notation $\CI_{j,n} (\weight,k)$. We first relate the proxies $\widehat p_n^{\full}$ and $\widehat p_n^{\splt}$ to the actual coverage probability $p_n$. For the full proxy, we have
\begin{align*}
\{
\CI_j^{\full}\subseteq\CI_{j,n}
\}
& = 
\{
\param_j\in\CI_{j,n},~\CI_j^{\full}\subseteq\CI_{j,n}
\}
\cup 
\{
\param_j\notin\CI_{j,n},~\CI_j^{\full}\subseteq\CI_{j,n}
\}.
\end{align*}
Hence,
\begin{align*}
\EE\widehat p_n^{\full}(\weight,k)
=
\PP (
\CI_j^{\full}\subseteq\CI_{j,n}
)
\leq 
\PP(
\param_j\in\CI_{j,n}
)
+ \PP (
\param_j\notin\CI_j^{\full}
)
\le
p_n(\weight,k)+ \gamma \alpha.
\end{align*}

To study the split proxy, note that
\[
\{
\param_j\in\CI_j^{\splt},
\ \param_j\notin\CI_{j,n}
\}
\subseteq
\{ 
\CI_j^{\splt}\not\subseteq\CI_{j,n}
\}.
\]
Conditional on $\query_j$, the set estimates $\CI_j^{\splt}$ and $\CI_{j,n}$ are independent. Hence,
\begin{align*}
\PP (
	\param_j\in\CI_j^{\splt},
	\ \param_j\notin\CI_{j,n}
	) 
& =
	\EE\left[
	\PP (
	\param_j\in\CI_j^{\splt}\mid\query_j
) \cdot
	\PP (
	\param_j\notin\CI_{j,n}\mid\query_j
)
	\right]\\
	&\ge
(1 - \gamma )
	\PP (
	\param_j\notin\CI_{j,n}
)
=
(1 - \gamma) [ 1-p_n(\weight,k) ].
\end{align*}
Therefore,
\begin{align*}
\EE\widehat p_n^{\splt}(\weight,k) =
\PP (\CI_j^{\splt} \subseteq\CI_{j,n}) = 1 - \PP (\CI_j^{\splt}\not\subseteq\CI_{j,n})
\leq 1 - (1 - \gamma) [ 1-p_n(\weight,k) ]
\end{align*}
and 
\[
p_n(\weight,k) \geq 
1 - \frac{
1 -	\EE\widehat p_n^{\splt}(\weight,k) 
}{
	1 - \gamma
}
=
\frac{
\EE\widehat p_n^{\splt}(\weight,k)  - \gamma
}{
1 - \gamma
}.
\]

\subsection{Proof of \Cref{thm-uniform-coverage}}\label{sec-thm-uniform-coverage-proof}
Let
\begin{equation*}
	\eta
	=
	\sqrt{
		\frac{\log [ (K \wedge M)/\delta ] }{2J}
	}
	\qquad\text{and}\qquad
	\alpha' = (1-\gamma) \alpha + \eta.
\end{equation*}
Define $\bar{k} (0) = K$ and
\begin{align}
	\bar{k} (\weight_m)
	=
	\max\left\{
	0 \leq k \leq \bar{k} (\weight_{m-1}) :~
	\EE \widehat p_n(\weight_m, i)\ge 1 - \alpha',~\forall i\in[k]
	\right\},
	\qquad m \in [M].
	\label{defn-k-bar}
\end{align}
Define
\begin{align}
	\Omega = \{ \weight_m :~ m \in [M] \text{ and } \bar{k} (\weight_m) < \bar{k} (\weight_{m-1}) \}.
\end{align}
By construction, $\{ \bar{k} (\weight)  \}_{\weight \in \Omega}$ are distinct elements in $\{ 0, 1, \cdots, K-1 \}$. Hence, $|\Omega| \leq K$. We get
\begin{align}
	|\Omega| \leq K \wedge M.
	\label{eqn-Omega-size}
\end{align}
We present two important facts about $\Omega$.
\begin{claim}\label{claim-Omega-bar}
	For any $\weight \in \Omega$, we have $\EE \widehat p_n(\weight,
	\bar{k}(\weight) + 1
	) < 1 - \alpha'$.
\end{claim}
\begin{proof}
	Choose any $\weight \in \Omega$. There exists $m \in [M]$ such that $\weight = \weight_m$. By definition, we have
	\begin{align*}
		\bar{k} (\weight_m)
		=
		\max U
		,\qquad\text{where}\qquad
		U =	\left\{
		0 \leq k \leq \bar{k} (\weight_{m-1}) :~
		\EE \widehat p_n(\weight_m, i)\ge 1 - \alpha' ,~\forall i\in[k]
		\right\}.
	\end{align*}
	Therefore, $\bar{k} (\weight_m) + 1 \notin U$. It follows from $\weight_{m} \in \Omega$ that $\bar{k} (\weight_m) + 1 \leq \bar{k} (\weight_{m-1})$. Hence, we must have $\EE \widehat p_n(\weight_m, 
	\bar{k} (\weight_m) + 1
	) < 1 - \alpha' $.
\end{proof}

\begin{claim}\label{claim-Omega-hat}
	If $\widehat{k} \leq \bar{k}$ holds on $\Omega$, then the inequality holds over the entire set $\Lambda$.
\end{claim}

\begin{proof}
	Choose any $\weight \in \Lambda$. If $\Omega = \varnothing$ or $\min \Omega > \weight$, then $\bar{k} (\weight) = \bar{k} (\weight_0) = K \geq \widehat{k} (\weight)$. 
	Otherwise, let $w = \max \{ \weight' \in \Omega :~ \weight' \leq \weight \}$. The construction of $\Omega$ forces $\bar{k} (\weight) = \bar{k}(w)$. Since $\widehat{k} \leq \bar{k}$ holds on $\Omega$, $\bar{k} (w) \geq \widehat{k}(w)$. The monotonicity of $\widehat{k}$ implies that $\widehat{k}(w) \geq \widehat{k}(\weight)$. Combining the above estimates gives $\bar{k} (\weight) \geq \widehat{k}(\weight)$.
\end{proof}

From Claim \ref{claim-Omega-hat}, we obtain that
\begin{align*}
	\PP \Big( \widehat{k} (\weight) \leq \bar{k} (\weight), ~ \forall \weight \in \Lambda \Big)
	\geq \PP \Big( \widehat{k} (\weight) \leq \bar{k}(\weight), ~ \forall \weight \in \Omega \Big)
	\geq 1 - \sum_{\weight \in \Omega } 
	\PP \Big( \widehat{k} (\weight) > \bar{k}(\weight)\Big).
\end{align*}
According to the definition of $\widehat{k}$ in \eqref{eq-frontier-est}, we have
\[ 
\{
\widehat{k} (\weight) > \bar{k}(\weight)
\}
=
\{
\widehat{k} (\weight) \geq \bar{k}(\weight) + 1
\}
\subseteq 
\{
\widehat p_n(\weight,
\bar{k}(\weight) + 1
) \geq 1-(1-\gamma)\alpha
\}
, \qquad \forall \weight \in \Omega
\]
Then, Claim \ref{claim-Omega-bar} forces 
\[ 
\{
\widehat{k} (\weight) > \bar{k}(\weight)
\}
\subseteq 
\{
\widehat p_n(\weight,
\bar{k}(\weight) + 1
) 
- \EE \widehat p_n(\weight,
\bar{k}(\weight) + 1
) 
\geq \eta
\}
, \qquad \forall \weight \in \Omega.
\]

Fix any $\weight \in \Omega$. By Assumption \ref{assumption-independence}, $\widehat p_n(\weight,
\bar{k}(\weight) + 1
)$ is the average of $J$ independent Bernoulli random variables. Hoeffding's inequality (Theorem 2.8 in \cite{BLM13}) implies that
\[
\PP \Big(
\widehat p_n(\weight,
\bar{k}(\weight) + 1
) 
- \EE \widehat p_n(\weight,
\bar{k}(\weight) + 1
) 
\geq \eta
\Big)
\leq 
e^{-2 J \eta^2}.
\]
By a union bound and \eqref{eqn-Omega-size},
\begin{align}
	\PP \Big( \widehat{k} (\weight) \leq \bar{k} (\weight), ~ \forall \weight \in \Lambda \Big)
	\geq 1 - |\Omega| e^{-2 J \eta^2}
	\geq 1 - (K \wedge M) e^{-2 J \eta^2} = 1 - \delta.
\end{align}
Denote by $\cA$ the event on the left-hand side. When $\cA$ happens, the definition of $\bar{k}$ in \eqref{defn-k-bar} and Assumption \ref{assumption-fallback} imply that
\[
\min_{
	\weight \in \Lambda,~
	0 \leq k \leq  \widehat{k} (\weight)
} 
\EE \widehat p_n(\weight, k) \geq 1 - \alpha'.
\]
By \Cref{lem-proxies},
\[
\min_{
	\weight \in \Lambda,~
	0 \leq k \leq  \widehat{k} (\weight)
} 
p_n(\weight, k) \geq 
\begin{cases}
	1 - \alpha - \eta &  \text{for the full-sample proxy}\\
	1 - \alpha - (1 - \gamma)^{-1} \eta & \text{for the split proxy}
\end{cases}.
\]

\subsection{Proof of \Cref{cor-marginal-coverage}}\label{sec-cor-marginal-coverage-proof}

Define the random variable
\[
p = \min_{
	\weight \in \Lambda,~
	0 \leq k \leq  \widehat{k} (\weight)
} 
\PP \Big(
\param_{\query}\in
\CI(\dataset_n,\syndataset_{  k}, \weight)
\Big).
\]
We have 
\begin{align}
\PP \Big(
\param_{\query}\in
\CI(\dataset_n,\syndataset_{\widetilde k},\widetilde\weight)
\Big)
\geq \EE p
= 1 - \alpha - \EE (1-\alpha-p )
.
\label{eqn-cor-coverage}
\end{align}
\Cref{thm-uniform-coverage} implies that for any $\delta \in (0,1)$, the following holds with probability at least $1 - \delta$:
\[
p
\geq 
	1 - \alpha - c 	\sqrt{
		\frac{\log [ (K \wedge M)/\delta ] }{2J}
	}.
\]
We obtain that
\[
\PP
(
1-\alpha-p >t
)
\le
\min \{
1,
(K \wedge M) e^{ - 2Jt^2 / c^2 }
\},
\qquad  \forall t\ge0.
\]
Define $u = c \sqrt{ \frac{\log (K \wedge M) }{2J} }$. We have
\begin{align*}
\EE (1-\alpha-p )
& \leq
\EE (1-\alpha-p )_+
= \int_{0}^{\infty} \PP
\Big(
(1-\alpha-p)_+>t
\Big) \rd t \\
& \leq 
	\int_0^\infty
	\min \{
	1,
	(K \wedge M) e^{ - 2Jt^2 / c^2 }
	\}
	\rd t
= u + \int_u^\infty
(K \wedge M) e^{ - 2Jt^2 / c^2 }
\rd t  \\
& \leq u + 
\int_u^\infty
\frac{t}{u} \cdot 
(K \wedge M) e^{ - 2Jt^2 / c^2 }
\rd t 
= u + 
\frac{c^2 (K \wedge M)  }{4 J u} 
\int_u^\infty
\rd (- e^{ - 2Jt^2 / c^2 } ) \\
& = u + \frac{c^2 (K \wedge M)  }{4 J u}  e^{-2 J u^2 / c^2}
= u + \frac{c^2 }{4 J u}  
= u \bigg(
1 + \frac{c^2 }{4 J u^2}  
\bigg)
\\& 
= 
c \sqrt{
\frac{
\log ( K \wedge M )
}{2J}
}
\bigg(
1
+\frac{1}{ 2  \log ( K \wedge M )} 
\bigg)
 \leq c \sqrt{
	\frac{
2 \log ( K \wedge M )
	}{J}
}.
\end{align*}
The last inequality follows from the fact that $2 \log ( K \wedge M ) \geq 2 \log 2 > 1$. Plugging the above bound into \eqref{eqn-cor-coverage} completes the proof.

{
\bibliographystyle{ims}
\bibliography{bib}

\begin{thebibliography}{25}
\expandafter\ifx\csname natexlab\endcsname\relax\def\natexlab#1{#1}\fi
\expandafter\ifx\csname url\endcsname\relax
  \def\url#1{\texttt{#1}}\fi
\expandafter\ifx\csname urlprefix\endcsname\relax\def\urlprefix{URL }\fi

\bibitem[{Abdel-Azim et~al.(2026)Abdel-Azim, Wang and Lin}]{AWL26}
\textsc{Abdel-Azim, A.}, \textsc{Wang, R.} and \textsc{Lin, X.} (2026).
\newblock Harnessing synthetic data from generative {AI} for statistical
  inference.
\newblock \textit{arXiv preprint arXiv:2603.05396} .

\bibitem[{Angelopoulos et~al.(2023)Angelopoulos, Bates, Fannjiang, Jordan and
  Zrnic}]{ABF23}
\textsc{Angelopoulos, A.~N.}, \textsc{Bates, S.}, \textsc{Fannjiang, C.},
  \textsc{Jordan, M.~I.} and \textsc{Zrnic, T.} (2023).
\newblock Prediction-powered inference.
\newblock \textit{Science} \textbf{382} 669--674.

\bibitem[{Bashari et~al.(2026)Bashari, Lee, Lotan, Dobriban and Romano}]{BLL26}
\textsc{Bashari, M.}, \textsc{Lee, Y.}, \textsc{Lotan, R.~M.},
  \textsc{Dobriban, E.} and \textsc{Romano, Y.} (2026).
\newblock General synthetic-powered inference.
\newblock In \textit{Forty-third International Conference on Machine Learning}.

\bibitem[{Bashari et~al.(2025)Bashari, Lotan, Lee, Dobriban and Romano}]{BLL25}
\textsc{Bashari, M.}, \textsc{Lotan, R.~M.}, \textsc{Lee, Y.},
  \textsc{Dobriban, E.} and \textsc{Romano, Y.} (2025).
\newblock Synthetic-powered predictive inference.
\newblock In \textit{The Thirty-ninth Annual Conference on Neural Information
  Processing Systems}.

\bibitem[{Boucheron et~al.(2013)Boucheron, Lugosi and Massart}]{BLM13}
\textsc{Boucheron, S.}, \textsc{Lugosi, G.} and \textsc{Massart, P.} (2013).
\newblock \textit{Concentration Inequalities: A Nonasymptotic Theory of
  Independence}.
\newblock Oxford University Press.

\bibitem[{Byun et~al.(2026)Byun, Gupta, Lipton, Childers and Wilder}]{BGL26}
\textsc{Byun, Y.}, \textsc{Gupta, S.}, \textsc{Lipton, Z.}, \textsc{Childers,
  R.} and \textsc{Wilder, B.} (2026).
\newblock Valid inference with imperfect synthetic data.
\newblock \textit{Advances in Neural Information Processing Systems}
  \textbf{38} 162430--162469.

\bibitem[{Haerpfer et~al.(2020)Haerpfer, Inglehart, Moreno, Welzel, Kizilova,
  Diez-Medrano, Lagos, Norris, Ponarin, Puranen et~al.}]{WVS20}
\textsc{Haerpfer, C.}, \textsc{Inglehart, R.}, \textsc{Moreno, A.},
  \textsc{Welzel, C.}, \textsc{Kizilova, K.}, \textsc{Diez-Medrano, J.},
  \textsc{Lagos, M.}, \textsc{Norris, P.}, \textsc{Ponarin, E.},
  \textsc{Puranen, B.} \textsc{et~al.} (2020).
\newblock World values survey: Round seven -- country-pooled datafile
  (2017-2020).
\newline\urlprefix\url{https://doi.org/10.14281/18241.1}

\bibitem[{Huang et~al.(2025)Huang, Wu and Wang}]{HWW25}
\textsc{Huang, C.}, \textsc{Wu, Y.} and \textsc{Wang, K.} (2025).
\newblock How many human survey respondents is a large language model worth? an
  uncertainty quantification perspective.
\newblock \textit{arXiv preprint arXiv:2502.17773} .

\bibitem[{Huang et~al.(2020)Huang, Stein, Rubin and Kou}]{HSR20}
\textsc{Huang, D.}, \textsc{Stein, N.}, \textsc{Rubin, D.~B.} and \textsc{Kou,
  S.} (2020).
\newblock Catalytic prior distributions with application to generalized linear
  models.
\newblock \textit{Proceedings of the National Academy of Sciences} \textbf{117}
  12004--12010.

\bibitem[{Iyengar et~al.(2025)Iyengar, Lin and Wang}]{ILW25}
\textsc{Iyengar, G.}, \textsc{Lin, Y.-S.~W.} and \textsc{Wang, K.} (2025).
\newblock Model-free assessment of simulator fidelity via quantile curves.
\newblock \textit{arXiv preprint arXiv:2512.05024} .

\bibitem[{Ji et~al.(2025)Ji, Pan, Zhu and Lei}]{JPL25}
\textsc{Ji, W.}, \textsc{Pan, Y.}, \textsc{Zhu, R.} and \textsc{Lei, L.}
  (2025).
\newblock Multi-armed bandits with machine learning-generated surrogate
  rewards.
\newblock \textit{arXiv preprint arXiv:2506.16658} .

\bibitem[{Keret and Shojaie(2025)}]{KSh25}
\textsc{Keret, N.} and \textsc{Shojaie, A.} (2025).
\newblock {GLM} inference with {AI}-generated synthetic data using misspecified
  linear regression.
\newblock \textit{arXiv preprint arXiv:2503.21968} .

\bibitem[{Lu et~al.(2026)Lu, Wang, Zhang and Zhang}]{LWZ26}
\textsc{Lu, C.}, \textsc{Wang, M.}, \textsc{Zhang, D.~J.} and \textsc{Zhang,
  H.} (2026).
\newblock Generative augmented inference of {LLM}-generated data for market
  research: Theory and empirical evidence.
\newblock \textit{arXiv preprint arXiv:2604.14575} .

\bibitem[{Ma and Zhang(2026)}]{MZh26}
\textsc{Ma, Z.} and \textsc{Zhang, A.~R.} (2026).
\newblock Synthetic augmentation in imbalanced learning: When it helps, when it
  hurts, and how much to add.
\newblock \textit{arXiv preprint arXiv:2601.16120} .

\bibitem[{Maurer and Pontil(2009)}]{MPo09}
\textsc{Maurer, A.} and \textsc{Pontil, M.} (2009).
\newblock Empirical {Bernstein} bounds and sample variance penalization.
\newblock \textit{Proceedings of the 22nd Annual Conference on Learning Theory}
  .

\bibitem[{Nakada et~al.(2024)Nakada, Xu, Li and Zhang}]{NXL24}
\textsc{Nakada, R.}, \textsc{Xu, Y.}, \textsc{Li, L.} and \textsc{Zhang, L.}
  (2024).
\newblock Synthetic oversampling: Theory and a practical approach using {LLM}s
  to address data imbalance.
\newblock \textit{arXiv preprint arXiv:2406.03628} .

\bibitem[{O'Hagan and Ro{\v{c}}kov{\'a}(2025)}]{ORo25}
\textsc{O'Hagan, S.} and \textsc{Ro{\v{c}}kov{\'a}, V.} (2025).
\newblock {AI}-powered {B}ayesian inference.
\newblock \textit{arXiv preprint arXiv:2502.19231} .

\bibitem[{OpenAI(2024)}]{GPT4o}
\textsc{OpenAI} (2024).
\newblock Hello {GPT-4o}.
\newline\urlprefix\url{https://openai.com/index/hello-gpt-4o/}

\bibitem[{OpenAI(2025)}]{GPT5}
\textsc{OpenAI} (2025).
\newblock {Introducing GPT-5}.
\newline\urlprefix\url{https://openai.com/index/introducing-gpt-5/}

\bibitem[{Shen et~al.(2023)Shen, Liu and Shen}]{SLS23}
\textsc{Shen, X.}, \textsc{Liu, Y.} and \textsc{Shen, R.} (2023).
\newblock Boosting data analytics with synthetic volume expansion.
\newblock \textit{arXiv preprint arXiv:2310.17848} .

\bibitem[{Tan and Zrnic(2026)}]{TZr26}
\textsc{Tan, L.} and \textsc{Zrnic, T.} (2026).
\newblock Valid inference with synthetic data via task exchangeability.
\newblock \textit{arXiv preprint arXiv:2606.13629} .

\bibitem[{Wang et~al.(2026)Wang, Zhang and Zhang}]{WZZ26}
\textsc{Wang, M.}, \textsc{Zhang, D.~J.} and \textsc{Zhang, H.} (2026).
\newblock Large language models for market research: A data-augmentation
  approach.
\newblock \textit{Marketing Science} \textbf{45} 728--751.

\bibitem[{Ye et~al.(2026)Ye, Lyu and Tao}]{YLT26}
\textsc{Ye, Z.}, \textsc{Lyu, J.} and \textsc{Tao, R.} (2026).
\newblock Allocating human oversight in {AI}-enabled analytics.
\newblock \textit{arXiv preprint arXiv:2604.12497} .

\bibitem[{Yin and Xin(2026)}]{YXi26}
\textsc{Yin, Q.} and \textsc{Xin, L.} (2026).
\newblock Synthetic but not infinite: How much {LLM}-generated data to use in
  market research.
\newblock \textit{Available at SSRN 6078686} .

\bibitem[{Zhao et~al.(2024)Zhao, Mondal, Tandon, Dillion, Gray and Gu}]{ZMT24}
\textsc{Zhao, W.}, \textsc{Mondal, D.}, \textsc{Tandon, N.}, \textsc{Dillion,
  D.}, \textsc{Gray, K.} and \textsc{Gu, Y.} (2024).
\newblock {W}orld{V}alues{B}ench: A large-scale benchmark dataset for
  multi-cultural value awareness of language models.
\newblock In \textit{Proceedings of the 2024 Joint International Conference on
  Computational Linguistics, Language Resources and Evaluation (LREC-COLING
  2024)}. ELRA and ICCL.

\end{thebibliography}
}

\end{document}